\documentclass[aps,letterpaper,twocolumn,nofootinbib,superscriptaddress]{revtex4}

\usepackage{amsfonts,amssymb,amsmath,amsthm} 
\usepackage{bbm}
\usepackage{mathrsfs}
\usepackage{hyperref}
\usepackage{hyphenat}

\theoremstyle{plain}
\newtheorem{theorem}{Theorem}
\newtheorem{lemma}{Lemma}
\newtheorem{corollary}{Corollary}

\theoremstyle{definition}
\newtheorem{definition}{Definition}

\begin{document}
\sloppy
\title{Quantum Private Information Retrieval has linear communication complexity}
\author{\"Amin \surname{Baumeler}}
\email{amin.baumeler@usi.ch}
\affiliation{Faculty of Informatics, Universit\`a della Svizzera italiana, Lugano, Switzerland}
\author{Anne \surname{Broadbent}}
\email{abroadbe@uottawa.ca}
\affiliation{Department of Mathematics and Statistics, University of Ottawa, Ottawa, Canada}

\date{\today}

\begin{abstract}
In Private Information Retrieval (PIR), a client queries an $n$-bit database in order to retrieve an entry of her choice, while maintaining privacy of her query value. Chor, Goldreich, Kushilevitz, and Sudan showed that, in the information-theoretical setting, a linear amount of communication is required for classical PIR protocols (thus the trivial protocol is optimal). This linear lower bound was shown by Nayak to  hold also in the quantum setting.  Here, we extend Nayak's result by considering approximate privacy, and requiring security only against \emph{specious} adversaries, which are, in analogy to classical honest-but-curious adversaries,  the weakest reasonable quantum adversaries.
We show that, even in this weakened scenario,  Quantum Private Information Retrieval (QPIR) requires~$n$ qubits of communication.
From this follows that Le Gall's recent QPIR protocol with sublinear communication complexity is not information-theoretically private, against the weakest reasonable cryptographic adversary.
\end{abstract}

\maketitle

\newcommand{\norm}[1]{\left\lVert#1\right\rVert}
\newcommand{\abs}[1]{\left| #1 \right|}
\newcommand{\id}{\mathbbm{1}}
\newcommand{\bra}[1]{\langle #1  |}
\newcommand{\ket}[1]{|  #1 \rangle}
\newcommand{\braket}[2]{\langle #1 | #2 \rangle}
\newcommand{\ketbra}[2]{\ket{#1}\bra{#2}}
\newcommand{\egoketbra}[1]{\ketbra{#1}{#1}}
\newcommand{\tr}{\mathrm{tr}}
\newcommand{\gate}[1]{\mathsf{#1}}
\newcommand{\bigO}[1]{\mathcal{O}\left(#1\right)}
\newcommand{\littleo}[1]{o\left(#1\right)}
\newcommand{\party}[1]{\mathscr{#1}}
\newcommand{\partyA}{\mathscr{A}}
\newcommand{\partyB}{\mathscr{B}}
\newcommand{\pureA}{\bar{\partyA}}
\newcommand{\pureB}{\bar{\partyB}}
\newcommand{\advA}{\tilde{\partyA}}
\newcommand{\spaceA}{\mathcal{A}}
\newcommand{\spaceAtilde}{\tilde{\mathcal{A}}}
\newcommand{\spaceApure}{\bar{\mathcal{A}}}
\newcommand{\spaceB}{\mathcal{B}}
\newcommand{\spaceBpure}{\bar{\mathcal{B}}}
\newcommand{\spaceR}{\mathcal{R}}
\newcommand{\spaceDB}{\mathcal{D}}
\newcommand{\spaceC}{\mathcal{C}}
\newcommand{\spaceS}{\mathcal{S}}
\newcommand{\spaceX}{\mathcal{X}}
\newcommand{\spaceY}{\mathcal{Y}}
\newcommand{\spaceCcompressed}{\mathcal{C'}}
\newcommand{\spaceAdvA}{\tilde{\spaceA}}
\newcommand{\operationA}{\mathscr{A}}
\newcommand{\operationAtilde}{\tilde{\mathscr{A}}}
\newcommand{\operationApure}{\bar{\mathscr{A}}}
\newcommand{\operationB}{\mathscr{B}}
\newcommand{\inputState}{\rho_\text{in}}
\newcommand{\simulator}{\mathscr{S}}
\newcommand{\QPIRprot}{\Pi_\gate{QPIR}}
\newcommand{\QPIRprotpure}{\QPIRprot(\pureA,\pureB)}

\section{Introduction}
The cryptographic scheme of \emph{Private Information Retrieval} (PIR) describes the problem of querying a database without suffering a loss in privacy. It was formally defined in 1998 by Chor, Goldreich, Kushilevitz, and Sudan~\cite{CKGS98}.
Intuitively, not losing privacy through a query means that the database server does not learn anything about the client's input.
An interesting question is:
How much communication does a PIR protocol require?
Sending the whole database to the client is a trivial PIR protocol, but it seems unsatisfactory with respect to the amount of communication.
Are there better solutions?
In the setting of one database server and information-theoretic privacy,  the trivial protocol is optimal (even against honest-but-curious adversaries).
This result was shown by Chor, Goldreich, Kushilevitz, and Sudan~\cite{CKGS98}.

Quantum computation and quantum communication, compared to the classical model, allows for improved solutions to cryptographic tasks~\cite{BB84,Shor97}.
A natural question thus arises: \emph{does quantum information allow PIR protocols with  sublinear communication complexity?} We call a PIR protocol where we make use of quantum resources a Quantum Private Information Retrieval (QPIR) protocol.
Nayak answered this question in the negative~\cite{Nay99},  with a proof sketch establishing a reduction to \emph{random access encoding}.\footnote{It has been claimed~\cite{Gall12} that Nayak proved a lower bound for two-message quantum protocols only, when in fact,  his claim encompasses protocols with arbitrary interaction. We attribute this misunderstanding to the succinctness of Nayak's original write-up (the result and proof are only a few sentences long). }
There are also other fields where allowing quantum computation and communication failed to qualitatively improve the classical result.
An example is bit commitment: as we know, perfect bit commitment is not possible in the classical setting; in the quantum setting it is also not possible~\cite{LC97,May97} (but see Chailloux and Kerenidis~\cite{CK11} for some  quantum improvements that \emph{are} possible).

Recently, Le Gall presented a QPIR protocol with a sublinear amount of communication~\cite{Gall12}. This result holds for a database that exactly follows the protocol specification.
Motivated by this seemingly contradictory result, we study here the communication complexity of QPIR protocols that are secure against  \emph{specious} adversaries. As defined by Dupuis, Nielsen, and Salvail~\cite{DNS10}, specious adversaries may deviate from the protocol, but only in a way that is essentially indistinguishable from the honest behaviour---they are a quantum analogue of classical honest-but-curious adversaries, thus corresponding to the  weakest reasonable cryptographic adversaries.

\paragraph{Main result} We show that, even in the case of approximate privacy and approximate correctness, QPIR against specious adversaries has  linear communication complexity.\footnote{This result has appeared as part of the  M.\,Sc.~thesis of one of the authors~\cite{Baumeler12}.} This establishes that the adversarial model in Le Gall's analysis does not fulfill the weakest reasonable security definition and closes the topic of single-server, information-theoretic QPIR.

\paragraph{Related Work}
 Nayak's lower bound was generalized by Jain, Radhakrishnan, and  Sen~\cite{JRS02}, who showed a trade-off for the loss in privacy between the client and database.
 QPIR has also been studied in the scenario of multiple servers~\cite{KdW04sym,KdW04}, in the scenario of symmetric privacy~\cite{KdW04sym},
as well as in the scenario where a cheating server is detected~\cite{GLM08}.
Recently, a practical symmetric QPIR protocol which is not unconditionally secure was developed~\cite{JSG10}.

\subsection{Specious adversaries}
\label{sec:introspecious}
We call a party which follows a protocol \emph{honest}.
A \emph{correct} protocol is a protocol that achieves its task, given that all the parties are honest.
Clearly, every meaningful protocol has to be correct.
If we now try to restrict the actions of an adversary as much as possible, we cannot violate the correctness requirement.
This means, the weakest adversary has to appear honest.

In classical cryptography there exists a notion of \emph{honest-but-curious} adversaries, which models the above description. Such adversaries follow the protocol (honesty), but record everything they see and try to extract a secret (curiosity).
A classical honest-but-curious adversary can do nothing more to break the privacy property of a protocol.

Dupuis, Nielsen, and Salvail followed this spirit and introduced the quantum analog to the honest-but-curious adversaries, and called them \emph{specious}~\cite{DNS10}. The honesty property, as well as the curiosity property cannot be translated one-to-one from the classical  to the quantum case.
To get a meaningful model, the definition needs to capture the essence of being the weakest adversary, as described above.

\looseness=-1
Attempting a translation from the classical  to the quantum case, we see that
a quantum adversary can also follow the protocol to be honest.
Curiosity means to copy everything the adversary sees and extracting a secret from it. In general, copying is not possible because of the no-cloning theorem~\cite{NOCLONING}. Therefore, a protocol can force a quantum honest-but-curious adversary to forget. This motivates the need for a security guarantee not only at the end of the protocol, but also \emph{during} the interaction.

Quantum adversaries, on the other hand, can act in a way indistinguishable from an honest party.
As an example, in some protocols it may be possible to delay measurements. This means, the adversary skips a measurement instruction and continues in superposition. At a later point in the protocol, if required, the adversary can perform the measurement, making it look like it was honest all the time.
In other words, at any step during the execution of the protocol, we specify that the adversary should be able to provide
 some state that, when joined with the state held by the honest party, is indistinguishable from the joint state of an honest interaction.
This is the essence of the definition of specious adversaries, which we define formally in Section~\ref{sec:specious-defn}.

 While the concept of specious adversaries is not yet in widespread use in the quantum cryptographic community, \emph{purification} attacks and the related \emph{purified adversaries} have long been known to present subtle challenges unique to the quantum world.
   Charles Bennett and Gilles Brassard  were among the first to draw attention to this type of attack, proposing a quantum bit commitment scheme, together with an explicit purification attack~\cite{BB84}.
  Purified adversaries (who can be seen as delaying their input choices by sending entangled states) are easily seen to be a special case of specious adversaries.

\subsection{Le Gall's QPIR protocol}
\label{sec:introlegall}
Recently, Le Gall presented a QPIR protocol with information-theoretic privacy~\cite{Gall12}.
His protocol achieves a communication complexity of $\bigO{\sqrt{n}}$, where~$n$ is the database size in bits.
At  first glance, this result seems to beat Nayak's lower bound of~$n$.
However, the price for this lower communication complexity is that the server must follow the protocol precisely.
Hence, Le Gall considers a different model of adversaries. One naturally wonders if such gains can be achieved by specious adversaries. Our main result (Theorem~\ref{thm:result}) rules out this possibility.

\section{Formal Model and Security Definitions}

In this section, we formally define our model and notions of correctness and privacy.  First, we give some basic notation.

\subsection{Notation}
\label{sec:notation}
We use calligraphic symbols to describe Hilbert spaces.
Subscripts of quantum states and quantum operations usually denote the associated Hilbert spaces.
Let~$\spaceA$ and~$\spaceB$ be two Hilbert spaces.
By~$\spaceA\otimes\spaceB$ we denote the joint Hilbert space.
The set~$L(\spaceA,\spaceB)$ is the set of all linear maps from~$\spaceA$ to~$\spaceB$.
The set~$L(\spaceA)=L(\spaceA,\spaceA)$ is the set of all linear maps on~$\spaceA$.
A quantum state is either expressed as a ket~$\ket{x}$ or as a density operator~$\rho$.
Every state than can be written as~$\ket{x}$ is pure, with corresponding density operator~$\ket{x}\bra{x}$.
The set~$S(\spaceA)$ is the set of all density operators on~$\spaceA$.
The identity operator on the space~$\spaceA$ is~$\id_\spaceA \in L(\spaceA)$,
for a joint space~$\spaceA\otimes\spaceB$ we use~$\id_{\spaceA,\spaceB}$.
An operator~$\gate{U} \in L(\spaceA)$ is called unitary, if~$\gate{U}^\dagger\gate{U}=\id$.
In the expression~$\spaceA \approx \spaceB$, the symbol~$\approx$ denotes that the dimension of~$\spaceA$ equals the dimension of~$\spaceB$ (i.e.,~$\dim(\spaceA) = \dim(\spaceB)$).
The measurement outcome of a measurement~$\mathcal{M}$ of a density operator~$\rho$ is expressed by~$\mathcal{M}(\rho)$.

Let~$\rho,\sigma \in S(\spaceA)$ be two density operators.
We denote by~$\norm{\rho}_1=\tr\abs{\rho}$ the \emph{trace norm} of the density operator~$\rho$.
The \emph{trace distance} between the two density operators~$\rho$ and~$\sigma$ is defined as
\begin{align}
  \Delta(\rho,\sigma) := \frac{1}{2} \norm{\rho-\sigma}_1
  \,.
\end{align}
If~$\rho=\ket{x}\bra{x}$ and~$\sigma=\ket{y}\bra{y}$ are pure, then we use the compact notation~$\Delta(\ket{x},\ket{y})$ interchangeably with~$\Delta(\ket{x}\bra{x},\ket{y}\bra{y})$.

\subsection{Protocol definition}
\label{sec:protocol}
As mentioned in Section~\ref{sec:introspecious}, when defining security against specious adversaries, we must examine the \nohyphens{system} held by the adversary \emph{during} the protocol. To this end, we first formally define a two-party quantum \emph{protocol}. We base our definition on   the strategy formalism of Gutoski and Watrous~\cite{GW07}, as well as on a definition from Dupuis, Nielsen and Salvail~\cite{DNS10} (our scenario is simpler since our protocols do not make any explicit oracle calls). Without loss of generality, we assume that party~$\partyA$ sends the first and last messages.

\begin{definition}[Two-party quantum protocol]
An
\emph{$s$-round, two-party protocol}
denoted
$\Pi=\left(\partyA,\partyB,s\right)$
consists of:
\begin{enumerate}
  \itemsep0em
  \item	input spaces~$\spaceA_0$ and~$\spaceB_0$ for parties~$\partyA$ and~$\partyB$ \nohyphens{respectively},
  \item	memory spaces $\spaceA_1,\dots,\spaceA_s$ for~$\partyA$ and $\spaceB_1,\dots,\spaceB_s$  for~$\partyB$ and communication spaces $\spaceX_1,\dots,\spaceX_s$, $\spaceY_1,\dots,\spaceY_{s-1}$,
  \item	an $s$-tuple of quantum operations $(\operationA_1,\dots,\operationA_s)$ for~$\partyA$, where $\operationA_1:L(\spaceA_{0}) \mapsto L(\spaceA_1 \otimes \spaceX_1)$, and   $\operationA_i:L(\spaceA_{i-1} \otimes \spaceY_{i-1}) \mapsto L(\spaceA_i \otimes \spaceX_i)$, $(2 \le i \le s)$,
  \item	an $s$-tuple of quantum operations $(\operationB_1,\dots,\operationB_s)$ for~$\partyB$, where $\operationB_i:L(\spaceB_{i-1} \otimes  \spaceX_i) \mapsto L(\spaceB_i \otimes \spaceY_{i})$, $(1\le i \le s-1)$, and $\operationB_s:L(\spaceB_{s-1} \otimes  \spaceX_s) \mapsto L(\spaceB_s)$.
\end{enumerate}
\end{definition}

If $\Pi=(\partyA,\partyB,s)$ is an $s$-round two-party protocol, we define the state after the  $i$-th step $(1 \leq i \leq 2s)$, and upon input state $\inputState \in S(\spaceA_0 \otimes \spaceB_0 \otimes \spaceR)$,
where $\spaceR$ is a system of dimension
\mbox{$\dim(\spaceR)=\dim(\spaceA_0)\dim(\spaceB_0)$},
as
\begin{align}
\rho_i\left(\inputState\right) :=
	(
    \operationA_{{(i+1)}/{2}} \otimes \id_{\spaceB_{{(i-1)}/{2}},\spaceR}
  )
  \ldots\notag\\ \left(
    \operationB_1 \otimes \id_{\spaceA_1,\spaceR}
  \right)
  \left(
    \operationA_1 \otimes \id_{\spaceB_0,\spaceR}
  \right)
  \left(
    \inputState
  \right)
  \,,
\end{align}
for~$i$ odd, and
\begin{align}
\rho_i\left(\inputState\right) :=
	(
    \operationB_{i/2} \otimes \id_{\spaceA_{i/2},\spaceR}
  )
  \ldots\notag\\ \left(
    \operationB_1 \otimes \id_{\spaceA_1,\spaceR}
  \right)
  \left(
    \operationA_1 \otimes \id_{\spaceB_0,\spaceR}
  \right)
  \left(
    \inputState
  \right)\,,
\end{align}
for~$i$ even.
Note that the last round (round~$s$) is only partial, since  $\operationB_s:L(\spaceB_{s-1} \otimes  \spaceX_s) \mapsto L(\spaceB_s)$. We define the final state of protocol $\Pi=(\partyA,\partyB,s)$,  upon input state $\inputState \in S(\spaceA_0 \otimes \spaceB_0 \otimes \spaceR)$ as:
\begin{align}
  &\left[\partyA \circledast \partyB\right]\left(\inputState\right)
  := \rho_{2s}\left(\inputState\right)  \,.
\end{align}
The \emph{communication complexity} of \mbox{$\Pi=(\partyA,\partyB,s)$} is the total amount of quantum communication in the protocol (counted in terms of qubits), as given by $\sum_{i=1}^s \log(\dim(\spaceX_i)) +  \sum_{i=1}^{s-1} \log(\dim(\spaceY_i))$.

\looseness=-1
Given a protocol $\Pi=\left(\partyA,\partyB,s\right)$, an \emph{adversary~$\advA$ for~$\partyA$} is an $s$-tuple of quantum operations $(\operationAtilde_1,\dots,\operationAtilde_s)$, where $\operationAtilde_1:L(\spaceA_{0}) \mapsto L(\spaceAtilde_1 \otimes \spaceX_1)$, and  \mbox{$\operationAtilde_i:L(\spaceAtilde_{i-1} \otimes \spaceY_{i-1}) \mapsto L(\spaceAtilde_i \otimes \spaceX_i)$}, $(2 \le i \le s)$, with $\spaceAtilde_1,\dots,\spaceAtilde_s$ being $\advA$'s memory spaces.
We denote the final state of a protocol run with an adversary~$\advA$ by~$[\advA \circledast \partyB]\left(\inputState\right)$.
The state after the~$i$-th step of a protocol run with an adversary~$\advA$ is~$\tilde{\rho}_i(\advA,\inputState)$.

A special type of adversary for a protocol \mbox{$\Pi=\left(\partyA,\partyB,s\right)$},   is a \emph{purified adversary,  $\pureA$ for~$\partyA$}  that is described by \emph{unitaries}  $(\operationApure_1,\dots,\operationApure_s)$, where $\operationApure_1:L(\spaceA_{0}\otimes \spaceApure_{0}) \mapsto L(\spaceA_1 \otimes \spaceApure_1 \otimes \spaceX_1)$ and $\operationApure_i:L(\spaceA_{i-1} \otimes \spaceApure_{i-1} \otimes \spaceY_{i-1}) \mapsto L(\spaceA_i \otimes \spaceApure_i \otimes \spaceX_i)$, $(2 \leq i \leq s)$, with auxiliary space $\spaceApure_{0}$ of sufficiently large dimension being initialized to the zero state. We refer to $\spaceApure_1, \ldots ,\spaceApure_s $ as the purifying spaces and specify that tracing out the purifying space reverts the state to a state from the original protocol; in particular, this holds for the final state of the protocol:  $\tr_{\spaceApure_s}[\pureA \circledast \partyB]\left(\inputState\right) =[\partyA \circledast \partyB]\left(\inputState\right)$ for all $\inputState$.
It is not hard to see that such adversaries always exist (see, for instance Gutoski and Watrous~\cite{GW07}).  The definition of a \emph{purified adversary,  $\pureB$ for~$\partyB$}
 is obtained similarly as the definition for~$\pureA$. In particular, $\Pi=\left(\pureA,\pureB,s\right)$ denotes the protocol where both parties $\partyA$ and $\partyB$ are purified.

\subsection{Specious adversaries}
\label{sec:specious-defn}
Recall  the intuition that a specious adversary should be able, at each step of the protocol, to produce a state that, when joined with the honest party's state, is close (in trace distance) to the joint state of an honest execution of the protocol. Dupuis, Nielsen, and Salvail~\cite{DNS10} give a definition for specious adversaries in the most general context. For the purposes of QPIR in our model, the following is an equivalent definition; below we also define \emph{ultimately specious} adversaries, which are adversaries that satisfy the criteria for speciousness at the last step of the protocol.

\begin{definition}[specious adversaries]
	Let~$\Pi=(\partyA,\partyB,s)$ be an~$s$-round two-party protocol.
	We say that an adversary  $\advA$ for~$\partyA$ is~\emph{$\varepsilon$-specious},
	if there exists a sequence of quantum operations $(\mathscr{F}_1, \dots, \mathscr{F}_{2s})$, such that for all~$1 \le i \le 2s$ and for all~$\inputState \in S(\spaceA_0 \otimes \spaceB_0 \otimes \spaceR)$,
		\begin{enumerate}
			\item  \begin{equation}\label{item:specious-1} \mathscr{F}_i:
\begin{cases} L(\spaceAdvA_{(i+1)/2}) \mapsto L(\spaceA_{(i+1)/2} ),& i \text{ even}\\
                L(\spaceAdvA_i \otimes \spaceX_{i/2 +1}) \mapsto L(\spaceA_i  \otimes \spaceX_{i/2+1}), & i \text{ odd}
 \end{cases}
\end{equation}
			\item  \label{item:specious-2}for every input state~$\inputState \in S(\spaceA_0 \otimes \spaceB_0 \otimes \spaceR)$,
			\begin{align}
				\Delta((\mathscr{F}_i \otimes \id_{\spaceB_i,\spaceR})(\tilde{\rho}_i(\advA,\inputState)), \rho_i(\inputState))
				\le \varepsilon
				\,.
			\end{align}
		\end{enumerate}
We call an adversary~$\advA$ for $\partyA$ \emph{ultimately $\varepsilon$-specious}, if there exists a quantum operation $\mathscr{F}: L(\spaceAdvA_{s} ) \mapsto L(\spaceA_{s} ) $, such that  for every input state~$\inputState \in S(\spaceA_0 \otimes \spaceB_0 \otimes \spaceR)$,
\begin{equation}\Delta((\mathscr{F} \otimes \id_{\spaceB_{s},\spaceR})(\left[\advA \circledast \partyB\right]\left(\inputState\right) ), 	\left[\partyA \circledast \partyB\right]\left(\inputState\right))
				\le \varepsilon
				\,.
\end{equation}
\end{definition}

\subsection{Definitions for QPIR}
Using the formalism developed so far, we now define QPIR protocols. In particular, we define a notion of approximate correctness, together with a notion of approximate privacy against specious servers;
\emph{correctness} refers to the notion that the client should obtain the correct outcome at the end of the protocol, while \emph{privacy} refers to the notion that the server should learn essentially nothing about the client's input via its interaction with the client. For specious adversaries, this  corresponds to the intuitive notion that the server's local density matrix at each step of the protocol should be independent of the client's input~$i$; in other words, there must exist at each step of the protocol, a quantum map $\simulator$ that has access only to the server's input register and that reproduces, or \emph{simulates} the server's local view. This is the standard ideal-real world simulation-based security notion, that is simplified to the QPIR setting and required only for specious adversaries.

We also consider \emph{ultimate} privacy (i.e., the privacy condition holds at the end of the protocol) against purified servers, which is sufficient in order to show our result.

\label{sec:QPIRprot}
\begin{definition}[QPIR protocol]
An \emph{$s$-round, $n$-bit Quantum Private Information Retrieval protocol}
is a two-party protocol
$\QPIRprot=\left(\partyA,\partyB,s\right)$,
where~$\partyA$ is the server and~$\partyB$ is the client.

We call~$\QPIRprot$ \emph{$(1-\delta)$-correct} if, for all inputs \mbox{$\inputState= \ket{x}\bra{x}_{\spaceA_0} \otimes \ket{i}\bra{i}_{\spaceB_0}$}, with $x = x_1,\ldots, x_n \in \{0,1\}^n$ and $i \in \{1, \ldots, n\}$, there exists a measurement~$\mathcal{M}$ with outcome~$0$ or $1$, such that:
\begin{equation}
\Pr[\mathcal{M}\left(\tr_{\spaceA_s}  	\left[\partyA \circledast \partyB\right]\left(\inputState\right)\right)  = x_i] \geq 1-\delta\,.
\end{equation}

  We call~$\QPIRprot$ \emph{$(1-\varepsilon)$-private against $\gamma$-specious servers} if
  for every $\gamma$-specious server $\advA$,
  there exists a sequence of quantum operations $\simulator_1, \ldots ,\simulator_{s-1}$ where
  \mbox{$\simulator_i : L(\spaceA_0) \mapsto L(\spaceAdvA_i \otimes \spaceY_{i})$}, such that for all $1 \leq i \leq s-1$  and for all~$\inputState \in S(\spaceA_0 \otimes \spaceB_0 \otimes \spaceR)$,
  \begin{align}
    \label{eq:speciousprivacy}
    \Delta\left( \tr_{\spaceB_0}\left( (\simulator_i \otimes \id_{\spaceB_0,\spaceR})(\inputState)\right), \tr_{\spaceB_i}\left(\tilde{\rho}_i(\advA,\inputState)\right) \right) \leq \varepsilon
    \,,
  \end{align}

  We call~$\QPIRprot$ \emph{ultimately $(1-\varepsilon)$-private against purified servers} if
  for every purification~$\pureA$ of the server~$\partyA$
  there exists a quantum operation
  $\simulator : L(\spaceA_0) \mapsto L(\spaceA_s \otimes \spaceApure_s)$, such that  for all~$\inputState \in S(\spaceA_0 \otimes \spaceB_0 \otimes \spaceR)$,
  \begin{align}
    \label{eq:privacy}
    \Delta\left(	\tr_{\spaceB_0} \left(
  	  \simulator\otimes\id_{\spaceB_0,\spaceR}
  	\right)
  	(\inputState),
   \tr_{\spaceB_s}
  	[\pureA \circledast \partyB]\left(\inputState\right)
    \right)
    \le
    \varepsilon
    \,.
  \end{align}
\end{definition}

\section{Tools}
\label{sec:preliminaries}
In this section, we present definitions and results that are used in the proof of our main result.

\subsection{Entropy}
\begin{definition}[Shannon entropy]
  Let $P_X$ be a probability distribution over the alphabet $X$. Then the Shannon entropy
  $H(P_X)$ of~$P_X$ is
  \begin{align}
    H(P_X) := -\sum_{x\in X} P_X(x) \log\left(P_X(x)\right)
    \,.
  \end{align}
\end{definition}
The Shannon entropy of a binary random variable is called \emph{binary entropy}:
\begin{definition}[binary entropy]
  Let~$p$ be the probability of an event of a binary random variable. Then the binary entropy~$H_\text{bin}(p)$ of~$p$ is
  \begin{align}
    H_\text{bin}(p) := -p\log p -(1-p)\log(1-p)
    \,.
  \end{align}
\end{definition}

\subsection{Trace distance and fidelity}

We have already encountered the \emph{trace norm} and \emph{trace distance} in Section~\ref{sec:notation}. Another measure of distance between two density operators~$\rho$ and~$\sigma$ is the \emph{fidelity}, defined as
\begin{align}
  F(\rho,\sigma) := \norm{\rho^{\frac{1}{2}}\sigma^{\frac{1}{2}}}_1
  \,.
\end{align}
For pure states~$\ket{x}$, and~$\ket{y}$ we define~$F(\ket{x},\ket{y})$ as~$F(\egoketbra{x},\egoketbra{y})$.
The following Lemma simplifies the calculation of the fidelity.
\begin{lemma}[Uhlmann's Lemma]
  \label{lemma:uhlmann}
  The fidelity between $\rho_\spaceA\in S(\spaceA)$ and $\sigma_\spaceA\in S(\spaceA)$ is
  \begin{align}
    F(\rho_\spaceA,\sigma_\spaceA)
    &=
    \max_{\ket{\varphi}_{\spaceA,\spaceB},\ket{\psi}_{\spaceA,\spaceB}}
    F\left(
      \ket{\varphi}_{\spaceA,\spaceB},
      \ket{\psi}_{\spaceA,\spaceB}
    \right)
    \\&=
    \max_{\ket{\varphi}_{\spaceA,\spaceB},\ket{\psi}_{\spaceA,\spaceB}}
    \abs{\braket{\varphi}{\psi}_{\spaceA,\spaceB}}
    \,,
  \end{align}
  where the maximum is taken over all purifications of $\rho_\spaceA=\tr_\spaceB\egoketbra{\varphi}_{\spaceA,\spaceB}$
  and over all purifications of \mbox{$\sigma_\spaceA=\tr_\spaceB\egoketbra{\psi}_{\spaceA,\spaceB}$}.
\end{lemma}
Recall the Fuchs-van de Graaf inequalities, relating the fidelity to the trace distance:
\begin{lemma}[Fuchs-van de Graaf inequalities]
  \label{lemma:fuchsvandegraaf}
  Let \mbox{$\rho,\sigma\in S(\spaceA)$} be density operators, then
  \begin{align}
    1-F(\rho,\sigma) \le \Delta(\rho,\sigma) \le \sqrt{1-F(\rho,\sigma)^2}
    \,.
  \end{align}
\end{lemma}

The following lemma states that, given two density matrices that are close in trace distance, it is possible, by acting only on the purifying subspace, to transform a purification of one of the density matrices into an approximate version of a purification of the other.

\begin{lemma}
  \label{lemma:uhlmannfuchsvandegraaf}
  Let $\rho_\spaceA,\sigma_\spaceA \in S(\spaceA)$ be two $\varepsilon$-close density operators, such that
  \begin{align}
    \Delta(\rho_\spaceA,\sigma_\spaceA) \le \varepsilon
  \end{align}
  with respective purifications~$\ket{\varphi}_{\spaceA,\spaceB}$ and~$\ket{\psi}_{\spaceA,\spaceB}$.
  Then there exists a unitary~$\gate{U}_\spaceB$ acting solely on~$\spaceB$, such that
  \begin{align}
    \Delta
    \left(
      \ket{\varphi}_{\spaceA,\spaceB},
      \left(
        \id_{\spaceA} \otimes \gate{U}_\spaceB
      \right)
      \ket{\psi}_{\spaceA,\spaceB}
    \right)
    \le
    \sqrt{\varepsilon(2-\varepsilon)}
    \,.
  \end{align}
\end{lemma}
\begin{proof}[Proof of Lemma~\ref{lemma:uhlmannfuchsvandegraaf}]
    Let $\rho_\spaceA,\sigma_\spaceA \in S(\spaceA)$ be two density operators, such that
    $\Delta(\rho_\spaceA,\sigma_\spaceA) \le \varepsilon$.
  By the first inequality of Lemma~\ref{lemma:fuchsvandegraaf} we get
  \begin{align}
  	\label{eq:fidlowerbound}
  	F(\rho_\spaceA,\sigma_\spaceA) \ge 1-\varepsilon
  	\,.
  \end{align}
  Let the state~$\ket{\varphi}_{\spaceA,\spaceB}$ be an arbitrary purification of the density operator~$\rho_\spaceA$. By Lemma~\ref{lemma:uhlmann}, there exists a purification~$\ket{\psi'}_{\spaceA,\spaceB}$ of the density operator~$\sigma_\spaceA$, such that
  \begin{align}
  	F(\rho_\spaceA,\sigma_\spaceA) = F(\ket{\varphi}_{\spaceA,\spaceB},\ket{\psi'}_{\spaceA,\spaceB})
  	\,.
  \end{align}
  Therefore, the fidelity lower bound~\eqref{eq:fidlowerbound} is also a lower bound for the fidelity between the pure states~$\ket{\varphi}_{\spaceA,\spaceB}$ and~$\ket{\psi'}_{\spaceA,\spaceB}$.
  Using this in the second inequality of Lemma~\ref{lemma:fuchsvandegraaf} yields
  \begin{align}
  	\Delta(\ket{\varphi}_{\spaceA,\spaceB},\ket{\psi'}_{\spaceA,\spaceB}) \le \sqrt{1-F(\ket{\varphi}_{\spaceA,\spaceB},\ket{\psi'}_{\spaceA,\spaceB})^2}
  	\,.
  \end{align}
  By squaring both sides and plugging in inequality~\eqref{eq:fidlowerbound} we get
  \begin{align}
  	\Delta(\ket{\varphi}_{\spaceA,\spaceB},\ket{\psi'}_{\spaceA,\spaceB})^2 &\le 1-F(\ket{\varphi}_{\spaceA,\spaceB},\ket{\psi'}_{\spaceA,\spaceB})^2\\
  	&\le 1-(1-\varepsilon)^2\\
  	&= \varepsilon (2-\varepsilon)
  	\,.
  \end{align}
  Because purifications are equivalent up to unitary transformations on the purifying system, we thus get
  \begin{align}
      \Delta
    \left(
      \ket{\varphi}_{\spaceA,\spaceB},
      \left(
        \id_\spaceA \otimes \gate{U}_\spaceB
      \right)
      \ket{\psi}_{\spaceA,\spaceB}
    \right)
    \le
    \sqrt{\varepsilon(2-\varepsilon)}
    \,,
  \end{align}
  where
  \begin{align}
  \left(
        \id_\spaceA \otimes \gate{U}_\spaceB
      \right)
      \ket{\psi}_{\spaceA,\spaceB}
      =
      \ket{\psi'}_{\spaceA,\spaceB}
      \,.
  \end{align}
\end{proof}

\subsection{The Schmidt decomposition and its properties}
The Schmidt compression allows for a lossless compression of a quantum state.
We first describe the Schmidt decomposition.
\begin{theorem}[Schmidt decomposition]
Let~$\ket{\psi}_{\spaceA,\spaceB}$ be a pure state shared between party~$\partyA$ and party~$\partyB$.
Then there exists a set of orthonormal pure states~$\{\ket{a_i}_\spaceA\}$ for party~$\partyA$, a set of orthonormal pure states~$\{\ket{b_i}_\spaceB\}$ for party~$\partyB$, a set of real coefficients~$\{\lambda_i\}$ called \emph{Schmidt coefficients}, and a positive integer~$r$ called \emph{Schmidt rank}, such that
\begin{align}
	\ket{\psi}_{\spaceA,\spaceB}
	=
	\sum_{i=1}^r
	\lambda_i
	\ket{a_i}_\spaceA
	\ket{b_i}_\spaceB
	\,.
\end{align}
\end{theorem}

Because the spaces~$\spaceA$ and~$\spaceB$ use only~$r$ different orthonormal pure states, both spaces can be compressed independently to spaces of dimension~$r$ with~$\lceil \log r\rceil$ qubits. This is known as \emph{Schmidt compression}.

The following theorem states that we can bound the Schmidt rank of a bipartite state resulting from a purified two-party protocol. This theorem is attributed to Kremer~\cite{Kre95} (see Lemma 5).

\begin{theorem}[bound on Schmidt rank]
	\label{thm:boundschmidtrank}
Let \mbox{$\Pi=\left(\pureA,\pureB,s\right)$} be a two-party quantum protocol with purified parties $\pureA$ and $\pureB$, and let $\inputState=\ket{\phi_0}_{\spaceA,\spaceB}$ be a pure product state. Suppose $\Pi$ has communication complexity~$c$. Then $[\pureA \circledast \pureB]\left(\inputState\right)$ has Schmidt rank at most~$2^c$.
\end{theorem}
\begin{proof}[Proof of Theorem~\ref{thm:boundschmidtrank}]

In the following, we ignore unitary operations on either side during the protocol, because such operations do not increase the Schmidt rank.

Let~$\ket{\phi_d}_{\spaceA,\spaceB}$ be the shared state after~$d$ qubits have been communicated and let
	\begin{align}
		\label{eq:dschmidtdecomposition}
		\ket{\phi_d}_{\spaceA,\spaceB}
		=
		\sum_{i=1}^{r}
		\lambda_i
		\ket{a_i}_\spaceA
		\ket{b_i}_\spaceB
	\end{align}
	be the corresponding Schmidt decomposition.
	The terms belonging to party~$\partyA$ from the Schmidt decomposition~\eqref{eq:dschmidtdecomposition} can be expanded as
	\begin{align}
		\label{eq:catstate}
		\ket{a_i}_\spaceA = \alpha_i \ket{a^0_i}_{\spaceA_\ell} \ket{0}_{\spaceA_r} + \beta_i \ket{a^1_i}_{\spaceA_\ell} \ket{1}_{\spaceA_r}
		\,.
	\end{align}
	Without loss of generality, assume that in the next step in the protocol, the qubit from the space~$\spaceA_r$ is sent from party~$\partyA$ to party~$\partyB$.
	By plugging in the expanded expression~\eqref{eq:catstate} into the Schmidt decomposition~\eqref{eq:dschmidtdecomposition}, we get
	\begin{align}
		\ket{\phi_d}_{\spaceA,\spaceB}
		=
		\sum_{i=1}^{r}
		\lambda_i
		\left(\alpha_i \ket{a^0_i}_{\spaceA_\ell} \ket{0}_{\spaceA_r} + \beta_i \ket{a^1_i}_{\spaceA_\ell} \ket{1}_{\spaceA_r}\right)
		\ket{b_i}_\spaceB
		\\=
		\sum_{i=1}^{r}
		\lambda_i
		\alpha_i \ket{a^0_i}_{\spaceA_\ell} \ket{0}_{\spaceA_r}
		\ket{b_i}_\spaceB
		+
		\lambda_i
		\beta_i \ket{a^1_i}_{\spaceA_\ell} \ket{1}_{\spaceA_r}
		\ket{b_i}_\spaceB
		\,.
	\end{align}
	Hence the transmission of one qubit at most doubles the number of summands, which is an upper bound of the Schmidt rank of the new Schmidt decomposition into the spaces~$\spaceA_\ell$ and~$\spaceA_r \otimes \spaceB$.
	By assumption, the initial state~$\ket{\phi_0}_{\spaceA,\spaceB}$ has Schmidt rank~$1$. Therefore, after communicating~$c$ qubits the Schmidt rank is at most~$2^c$.
\end{proof}
\subsection{Random access encoding}

A random access encoding is an encoding of classical database as a density operator, such that any database item can be extracted with a certain probability using a measurement which is independent of the database.
It is easy to see that the message of a single-message QPIR protocol is a random access encoding of the server's database.
We state the definition of random access encoding and a theorem on their size; here, we consider the average case scenario, which follows from Nayak's work~\cite{Nay99} (see also Kerenidis and de Wolf~\cite{KdW04}, Appendix~B).
\begin{definition}[Random Access Encoding]
  An \mbox{\emph{$(n,m,p)$-random access encoding}} is a function~$f$ that maps~$n$-bit strings to density operators over~$m$ qubits, such that,
  for every $i\in\left\{ 1,\dots,n \right\}$, there exists a measurement $\party{M}_i$ with outcome~$0$ or~$1$ that has the property that on average over all $x\in\left\{ 0,1 \right\}^n$,
  \begin{align}
    \Pr\left[ \party{M}_i\left( f\left( x \right) \right)=x_i \right] \ge p
    \,.
  \end{align}
\end{definition}

\begin{theorem}[size of Random Access Encoding]
  \label{thm:rac}
  Any $(n,m,p)$-random access encoding satisfies \mbox{$m\ge\left( 1-H_\text{bin}\left( p \right) \right)n$}.
\end{theorem}

\section{Main Theorem}
\newcommand{\result}{
\begin{align}
  \left(1-H_\text{bin}\left(1-\delta-2\sqrt{\varepsilon(1-\varepsilon)}\right)\right)n\,.
\end{align}
}
\label{sec:lowerbound}

In this section, we present our main result  and related corollaries. The proof is given in Section~\ref{sec:proof}.

\subsection{Results}
\label{sec:results}

Our main result is the following.

\begin{theorem}
  \label{thm:result}
  Let $\QPIRprot=\left(\partyA,\partyB,s\right)$ be an $s$-{round}, $n$-bit QPIR protocol, that is $(1-\delta)$-correct and ultimately \mbox{$(1-\varepsilon)$}-private against purified servers. Then~$\QPIRprot$ has communication complexity of at least
  \result
\end{theorem}
The above theorem is an extension of Nayak's result on QPIR~\cite{Nay99} to approximate privacy, and requiring security only against a purified server at the end of the protocol.
It is easy to see that a purified server is specious (see Section~\ref{sec:protocol}).
Therefore, any QPIR protocol that is \mbox{$(1-\varepsilon)$}-private against \mbox{$\gamma$-specious} servers is also \mbox{$(1-\varepsilon)$}-private against purified serves.
Trivially such a protocol is ultimately \mbox{$(1-\varepsilon)$}-private against purified servers.
Hence, by Theorem~\ref{thm:result} we get.
\begin{corollary}
\label{cor:specious}
  Let $\QPIRprot=\left(\partyA,\partyB,s\right)$ be an $s$-{round}, $n$-bit QPIR protocol that is  $(1-\delta)$-correct and \mbox{$(1-\varepsilon)$}-private against \mbox{$\gamma$-specious}  servers. Then for any~$\gamma$, $\QPIRprot$ has communication complexity of at least
  \result
\end{corollary}

Let $\delta$ and $\varepsilon$ be nonnegative and negligible functions\footnote{A nonnegative function $\mu$ is called \emph{negligible with respect to $n$} if for all $c > 0$ and all sufficiently large $n$, $\mu(n) < n^{-c}$.} with respect to~$n$. Then for any~$\gamma$,  the communication complexity as given in Corollary~\ref{cor:specious} at least $n-\littleo{1}$. In sharp contrast to this,
in Le Gall's model (that considers an adversary that follows the protocol exactly), the communication complexity is~$\bigO{\sqrt{n}}$; we therefore obtain the following corollary.
\begin{corollary} \label{coro:LeGall}
  Le Gall's QPIR protocol is not private against $\gamma$-specious adversaries, for any~$\gamma$.
\end{corollary}
An alternate proof of Corollary~\ref{coro:LeGall}, via an explicit specious attack, can be found in the thesis of \"A.\,B.~\cite{Baumeler12}.

\subsection{Proof of Theorem~\ref{thm:result}}
\label{sec:proof}
The main technique used in the proof of Theorem~\ref{thm:result} is to reduce a given  QPIR protocol to a random access encoding, and then apply Nayak's lower bound as established by Theorem~\ref{thm:rac}. This is the same technique as used by Nayak in his lower bound proof for QPIR, which we extend here to the case of approximate privacy against ultimately specious servers.

As a  starting point to understanding the reduction, note that any single-message QPIR protocol (where one message is sent from the server to the client) implements a random access encoding.
Hence, the lower bound on the size of the random access encoding is also a lower bound on the communication complexity for the single-message QPIR protocol.
We generalize this idea to ultimately $(1-\varepsilon)$-private against purified servers, multi-round QPIR protocols by reducing the multi-round protocol to a single-message protocol, and hence to a random access encoding. Taking care that this procedure does not increase the amount of communication allows us to apply the lower bound on the size of the random access encoding to the communication complexity of the multi-step QPIR protocol, thus establishing the result.

\begin{proof}[Proof of Theorem~\ref{thm:result}]
Let $\QPIRprot$ be an  $s$-round, $n$-bit, $(1-\delta)$-correct Quantum Private Information Retrieval protocol that is ultimately $(1-\varepsilon)$-private against purified servers and that has communication complexity~$c$.

Consider $\QPIRprotpure$, the modification of~$\QPIRprot$, where both parties,~$\partyA$ and~$\partyB$, are purified, as described in Section~\ref{sec:protocol}.
We denote by~$\spaceS\approx\spaceA_s\otimes\spaceApure_s$ the server's subspace, and by~$\spaceC\approx\spaceB_s\otimes\spaceBpure_s$ the client's subspace at the end of the protocol. Furthermore, let
\begin{equation}
  \egoketbra{\bar{\psi}_{x,i}}_{\spaceS,\spaceC} := [\pureA \circledast \pureB]\left(\ket{x}\bra{x}\otimes \ket{i}\bra{i}\right)\,;
  \end{equation}
that is,  $\ket{\bar{\psi}_{x,i}}_{\spaceS,\spaceC}$ is the global state at the end of the protocol~$\QPIRprotpure$, with inputs $x\in\left\{ 0,1 \right\}^n$ for the database and
$i\in\left\{ 1,\dots,n\right\}$ for the index.

\textbf{Encoding.}  Given~$\QPIRprotpure$, we derive a random access encoding  in the following way:
 the server simulates the purified version $\QPIRprotpure$ of the protocol $\QPIRprot$ with inputs $\ket{x}$ as database input  and index $\ket{i}=\ket{1}$.
The joint output is
	  $\ket{\bar{\psi}_{x,1}}_{\spaceS,\spaceC}$.

Consider $\ket{\xi}_\spaceDB$, the uniform superposition of all possible databases
\begin{align}
 \ket{\xi}_\spaceDB
  :=
  \frac{1}{\sqrt{2^n}}
  \sum_{x\in\left\{0,1\right\}^n}\ket{x}_\spaceDB\,,
\end{align}
and let
  \begin{align}
    \egoketbra{\nu_i}_{\spaceS,\spaceC}
    :&=
    [\pureA \circledast \pureB]\left(\egoketbra{\xi}_\spaceDB \otimes \egoketbra{i}\right)
    \,.
  \end{align}
  Since we consider the case where both parties are purified, the final global state is
  \begin{align}
    \ket{\nu_i}_{\spaceS,\spaceC}=
    \frac{1}{\sqrt{2^n}}
    \sum_{x\in\left\{ 0,1 \right\}^n}
    \ket{\bar{\psi}_{x,i}}_{\spaceS,\spaceC}
    \,.
  \end{align}

By Theorem \ref{thm:boundschmidtrank}, the Schmidt decomposition of $ \ket{\nu_i}_{\spaceS,\spaceC}$ into the subspace~$\spaceS$ and $\spaceC$ has Schmidt rank at most~$2^c$. Hence there exists a Schmidt compression of the subspace~$\spaceC$ into at most~$c$ qubits. By linearity, this map can be used to compress (and decompress)   $\ket{\bar{\psi}_{x,1}}_{\spaceS,\spaceC}$ for any $x \in \{0,1\}^n$. The server applies this compression on system $\spaceC$ of  $\ket{\bar{\psi}_{x,1}}_{\spaceS,\spaceC}$.
 Let the result of the compression be~$\ket{\bar{\psi}^c_{x,1}}_{\spaceS,\spaceCcompressed}$.
The server outputs as encoding of database~$x$ the state of the subsystem $\spaceCcompressed$:
\begin{align}
	  \tr_\spaceS\egoketbra{\bar{\psi}^c_{x,1}}_{\spaceS,\spaceCcompressed}
	  \,.
	\end{align}

\textbf{Decoding.} Given the output of the Encoding algorithm, the client applies the inverse operation of the Schmidt compression obtained above in order to recover the joint state corresponding to the
input $i=1$:
	\begin{align}
	  \ket{\bar{\psi}_{x,1}}_{\spaceS,\spaceC}
	  \,.
	\end{align}
However, the client would like to recover the joint state for an arbitrary $i$.
To this end, consider again $\ket{\xi}_\spaceDB$, the uniform superposition of databases as database input, and fix $i \in \{1, \ldots, n\}$ as index input. Let the corresponding input state be~$\inputState^{\xi,i}$. By the privacy condition (Equation~\ref{eq:privacy}), there exists a quantum map \mbox{$\simulator : L(\spaceA_0) \mapsto L(\spaceA_s \otimes \spaceApure_s)$}, such that
  \begin{equation}
  \label{eqn:privacy-proof}
  \Delta\left(	\tr_{\spaceB_0} \left(
  	  \simulator\otimes\id_{\spaceB_0,\spaceR}
  	\right)
  	(\inputState^{\xi,i}),
   \tr_{\spaceB_s}
  	[\pureA \circledast \partyB]\left(\inputState^{\xi,i}\right)
    \right)
    \le
    \varepsilon \,.
  \end{equation}
Since for all  $i \in \{1, \ldots, n\}$      \begin{equation}\tr_{\spaceB_0} \left(
  	  \simulator\otimes\id_{\spaceB_0,\spaceR}
  	\right)
  	(\inputState^{\xi,1}) = \tr_{\spaceB_0} \left(
  	  \simulator\otimes\id_{\spaceB_0,\spaceR}
  	\right)
  	(\inputState^{\xi,i})
  \end{equation}
and
\begin{align}
	  \tr_{\spaceB_s} [\pureA \circledast \partyB]\left(\inputState^{\xi,i}\right) &=   \tr_{\spaceC} [\pureA \circledast \pureB]\left(\inputState^{\xi,i}\right) \\&= \tr_{\spaceC}
      \egoketbra{\nu_i}_{\spaceS,\spaceC}\,,
\end{align}
  by the triangle inequality, we get that for all \mbox{$i \in \{1, \ldots, n\}$},
  \begin{equation}\Delta\left(
      \tr_{\spaceC}
      \egoketbra{\nu_1}_{\spaceS,\spaceC},
      \tr_{\spaceC}
      \egoketbra{\nu_i}_{\spaceS,\spaceC}
    \right)
    \le
    2\varepsilon
    \,.
\end{equation}

Thus by Lemma~\ref{lemma:uhlmannfuchsvandegraaf}, for every $i\in\left\{1,\dots,n\right\}$, there exists a unitary $\gate{U}^{1\to i}_\spaceC$ acting only on the client's subspace, such that
  \begin{align}
    \label{proofof:lowerbound.uisgood}
    \Delta
    \left(
    \left(
    \id_\spaceS
    \otimes
    \gate{U}^{1\to i}_\spaceC
    \right)
    \ket{\nu_1}_{\spaceS,\spaceC}
    ,
    \ket{\nu_i}_{\spaceS,\spaceC}
    \right)
    \le
    2 \sqrt{\varepsilon(1-\varepsilon)}
    \,.
  \end{align}
  Because the trace distance does not increase under measurements, we simply measure the space~$\spaceDB$ of the states from the inequality~\eqref{proofof:lowerbound.uisgood}
  and obtain that for a uniform random $x\in\left\{ 0,1 \right\}^n$
  \begin{align}
    \Delta
    \left(
    \left(
    \id_\spaceS
    \otimes
    \gate{U}^{1\to i}_\spaceC
    \right)
    \ket{\bar{\psi}_{x,1}}_{\spaceS,\spaceC}
    ,
    \ket{\bar{\psi}_{x,i}}_{\spaceS,\spaceC}
    \right)
    \le
    2 \sqrt{\varepsilon(1-\varepsilon)}
    \,.
  \end{align}
  Hence, on average over all databases $x\in\left\{ 0,1 \right\}^n$, this family $\left\{\gate{U}^{1\to i}_\spaceC\right\}_i$ of unitary operators can be used to construct a
  $2 \sqrt{\varepsilon(1-\varepsilon)}$-close approximation.

   It remains to calculate the recovery probability of the constructed random access code.
  The QPIR protocol~$\QPIRprot$ is~$(1-\delta)$-correct and hence there exists a measurement that recovers the desired bit with a probability of at least~$1-\delta$.
  The family of unitary approximation transformations~$\left\{ \gate{U}^{1\to i}_\spaceC \right\}_i$, used to approximate the global state, induces a loss in the recovery probability.  The approximation is $2\sqrt{\varepsilon(1-\varepsilon)}$-close.

  Hence  the QPIR protocol yields a  random access encoding with recovery probability of at least
  \mbox{$1-\delta-2\sqrt{\varepsilon(1-\varepsilon)}$}.
  By applying Nayak's Theorem~\ref{thm:rac}, we get that
  any $n$-bit, $(1-\delta)$-correct, ultimately \mbox{$(1-\varepsilon)$-private} against purified servers QPIR protocol has communication complexity of at least
  \begin{align}
    \left(1-H_\text{bin}\left( 1 - \delta - 2\sqrt{\varepsilon(1-\varepsilon)} \right)\right)n
    \,.
  \end{align}
\end{proof}

It is interesting to note that the reason why this lower bound proof is not applicable to the model in the work of Le Gall~\cite{Gall12}, is that there the privacy condition~\eqref{eq:privacy} does not hold.
In other words, the \emph{possibility} of Le Gall's result is a direct consequence of the fact that security is guaranteed only for classical inputs, that is, the adversary is forced to select a classical database at the beginning of the protocol, or equivalently, is forced to measure any superposition of databases that it might receive as input.

\section{Conclusion and Open Questions}
\label{sec:conclusion}
Using quantum computation and quantum communication, non-trivial information-theoretic single-server QPIR protocols secure against any reasonable adversary do not exist.
This work closes the topic of single-server and information-theoretic QPIR.

An open question that remains, is whether there exist other applications of the reduction from multi-step protocols to single-step protocols used in the proof of the lower bound (see Section~\ref{sec:lowerbound}).
In the reduction, we show that any protocol with asymmetric privacy at the end of the protocol against one particular type of adversaries, can be transformed to a single-step protocol.
The resulting single-step protocol preserves the communication complexity and the privacy property.
This reduction could potentially be used to build offline protocols from a multi-step protocols.
An offline protocol is a protocol, where the parties are not required to be involved in the protocol at the same time.
This could be advantageous under some circumstances.

\section{Acknowledgements}
\label{sec:ack}
We are grateful to Gus Gutoski, Robert K\"onig and Ashwin Nayak for helpful discussions, and to the anonymous referees for helpful comments. Furthermore, we thank S\'ebastien Gambs for introducing us to~PIR.\@  \"A.\,B.~thanks the Institute for Quantum Computing (IQC) and the University of Waterloo for hosting him for a six-month visit, during which these results were established. This work was performed while~A.\,B.~was at the Department of Combinatorics and Optimization, and at the Institute for Quantum Computing (IQC), University of Waterloo.
This work was supported by the Canadian Institute for Advanced Research (CIFAR), NSERC Frequency and Industry Canada.

\bibliographystyle{plain}
\bibliography{ref}

\begin{thebibliography}{10}

\bibitem{Baumeler12}
\"Amin Baumeler.
\newblock Quantum private information retrieval.
\newblock Master's thesis, ETH Z\"urich, 2012.
\newblock Available at:
  \url{http://e-collection.library.ethz.ch/view/eth:6297}.

\bibitem{BB84}
Charles~H. Bennett and Gilles Brassard.
\newblock Quantum cryptography: Public key distribution and coin tossing.
\newblock In {\em Proceedings of the International Conference on Computers,
  Systems, and Signal Processing}, pages 175--180, 1984.

\bibitem{CK11}
Andr{\'e} Chailloux and Iordanis Kerenidis.
\newblock Optimal bounds for quantum bit commitment.
\newblock In {\em Proceedings of the 52th Annual Symposium on Foundations of
  Computer Science}, FOCS `11, pages 354--362, 2011.

\bibitem{CKGS98}
Benny Chor, Eyal Kushilevitz, Oded Goldreich, and Madhu Sudan.
\newblock Private information retrieval.
\newblock {\em Journal of the ACM}, 45(6):965--981, Nov 1998.

\bibitem{DNS10}
Fr\'{e}d\'{e}ric Dupuis, Jesper~Buus Nielsen, and Louis Salvail.
\newblock Secure two-party quantum evaluation of unitaries against specious
  adversaries.
\newblock In {\em Proceedings of the 30th Annual Conference on Advances in
  Cryptology}, CRYPTO `10, pages 685--706. Sprinter-Verlag, 2010.

\bibitem{GLM08}
Vittorio Giovannetti, Seth Lloyd, and Lorenzo Maccone.
\newblock Quantum private queries.
\newblock {\em Physical Review Letters}, 100(23):230502, Jun 2008.

\bibitem{GW07}
Gus Gutoski and John Watrous.
\newblock Toward a general theory of quantum games.
\newblock In {\em Proceedings of the 39th Annual ACM Symposium on Theory of
  Computing}, STOC `07, pages 565--574, 2007.

\bibitem{JRS02}
Rahul Jain, Jaikumar Radhakrishnan, and Pranab Sen.
\newblock A property of quantum relative entropy with an application to privacy
  in quantum communication.
\newblock {\em Journal of the ACM}, 56(6):33, Sep 2009.
\newblock Preliminary version in FOCS `02.

\bibitem{JSG10}
Markus Jakobi, Christoph Simon, Nicolas Gisin, Jean-Daniel Bancal, Cyril
  Branciard, Nino Walenta, and Hugo Zbinden.
\newblock Practical private database queries based on a
  quantum-key-distribution protocol.
\newblock {\em Physical Review A}, 83:022301, Feb 2011.

\bibitem{KdW04}
Iordanis Kerenidis and Ronald de~Wolf.
\newblock Exponential lower bound for 2-query locally decodable codes via a
  quantum argument.
\newblock {\em Journal of Computer and System Sciences}, 69(3):395--420, Nov
  2004.

\bibitem{KdW04sym}
Iordanis Kerenidis and Ronald de~Wolf.
\newblock Quantum symmetrically-private information retrieval.
\newblock In {\em Information Processing Letters}, volume~90, pages 109--114,
  2004.

\bibitem{Kre95}
Ilan Kremer.
\newblock Quantum communication.
\newblock Master's thesis, The Hebrew University of Jerusalem, Mar 1995.
\newblock Available at: \url{http://www.cs.huji.ac.il/~noam/kremer-thesis.ps}.

\bibitem{Gall12}
Fran\c{c}ois Le~Gall.
\newblock Quantum private information retrieval with sublinear communication
  complexity.
\newblock {\em Theory of Computing}, 8(1):369--374, Jul 2012.

\bibitem{LC97}
Hoi-Kwong Lo and Hoi~Fung Chau.
\newblock Is quantum bit commitment really possible?
\newblock {\em Physical Review Letters}, 78(17):3410--3413, Apr 1997.

\bibitem{May97}
Dominic Mayers.
\newblock Unconditionally secure quantum bit commitment is impossible.
\newblock {\em Physical Review Letters}, 78(17):3414--3417, Apr 1997.

\bibitem{Nay99}
Ashwin Nayak.
\newblock Optimal lower bounds for quantum automata and random access codes.
\newblock In {\em Proceedings of the 40th Annual Symposium on Foundations of
  Computer Science}, FOCS `99, pages 369--376, 1999.

\bibitem{Shor97}
Peter Shor.
\newblock Polynomial-time algorithms for prime factorization and discrete
  logarithms on a quantum computer.
\newblock {\em SIAM Journal on Computing}, 26(5):1484--1509, Nov 1997.

\bibitem{NOCLONING}
William~K. Wootters and Wojciech~H. Zurek.
\newblock A single quantum cannot be cloned.
\newblock {\em Nature}, 299(5886):802--803, Oct 1982.

\end{thebibliography}

\end{document}